\documentclass[letterpaper,10pt,conference]{ieeeconf}
\IEEEoverridecommandlockouts
\let\labelindent\relax
\usepackage{enumitem,graphicx,xspace}
\usepackage{subfig}

\usepackage{amsmath,amssymb, amsfonts,cases,mathtools} 
\usepackage{dsfont}
\newcommand{\pde}[2]{ \frac{\partial #1}{\partial #2} }
\newcommand{\ppde}[2]{\frac{\partial^2 #1}{\partial #2^2}}

\newcommand{\R}{\mathbb{R}}

\newcommand{\1}{\mathds{1}}

\newcommand{\diag}[1]{\mathrm{diag}\left(#1\right)}

\newcommand{\card}[1]{\mathrm{card}(#1)}

\newcommand{\abs}[1]{ \left\lvert #1 \right\rvert }

\newcommand{\G}{\mathcal{G}}
\newcommand{\edgeSet}{\mathcal{E}}
\newcommand{\nodeSet}{\mathcal{V}}

\newcommand{\vspan}[1]{\mathrm{span}(#1)}
\newcommand{\U}{\mathcal{U}}

\newtheorem{theorem}{Theorem}
\newtheorem{corollary}{Corollary}
\newtheorem{conjecture}{Conjecture}
\newtheorem{lemma}{Lemma}
\newtheorem{assumption}{Assumption}
\newtheorem{remark}{Remark}

\usepackage[colorlinks=true,citecolor=blue,linkcolor=blue,urlcolor=blue,]{hyperref}

\usepackage{comment}

\newcommand{\xeq}{x^\ast}
\newcommand{\uast}{u^\ast}

\newcommand{\xtarget}{x_\mathrm{d}}

\newcommand{\opt}[1]{#1^{\text{opt}}}

\usepackage{pgfplots}
\pgfplotsset{compat=1.18}
\usetikzlibrary{pgfplots.fillbetween}
\usepackage{kthcolors}

\allowdisplaybreaks

\title{\LARGE\bf On personal recommendations in social networks}

\author{
Philipp Grünter$^{\ast}$\quad
Karl Henrik Johansson\quad
Angela Fontan
\thanks{
P. Grünter, K. H. Johansson, and A. Fontan are with the Department of Decision and Control Systems, KTH Royal Institute of Technology, Stockholm, Sweden. 
E-mail: \{grunter, kallej, angfon\}@kth.se. 
The authors are also affiliated with Digital Futures and WASP. $^\ast$Corresponding author.
}
\thanks{This work was partially supported by the Wallenberg AI, Autonomous Systems and
Software Program (WASP) funded by the Knut and Alice Wallenberg Foundation.}
}

\begin{document}

\maketitle
\begin{abstract}
Social networks in which algorithms actively influence humans through personal recommendations are ubiquitous. While opinion dynamics is an established tool to analyze these systems, existing models typically do not capture how individual agents process personal recommendations. In this work, we introduce a model for personal recommendations that is analytically tractable and consistent with the confirmation bias phenomenon from behavioral psychology. We describe how individuals process recommendations based on prior beliefs and a sensitivity parameter using a Gaussian influence function. Using this model, we analyze the effect of different recommendation policies. For broadcast policies, where recommendations are homogeneous across the agents, a bifurcation analysis shows that the system exhibits bistability, which may result in unintended consequences. For personally targeted policies, we first derive optimal personal recommendations and then extend to robust convergence when the agents' sensitivity to personal recommendations is uncertain. The collective behavior of the proposed model and the derived policies are illustrated and evaluated through numerical simulations. 
\end{abstract}

\section{Introduction}\label{sec:intro}

Personal recommendations are to a large extent generated by algorithms that automatically curate suggestions for individual users. These automated recommendations increasingly influence and shape human communities in domains such as transportation \cite{annaswamyEquitableMicrotransitServices2024}, energy systems \cite{ratliffIncentiveDesignUtility2014}, and social networks~\cite{lanzettiImpactRecommendationSystems2023}. One of the central challenges in these systems is modeling how individuals respond to such influence and understanding how a centralized target impacts collective behavior. In this context, personalized recommendations have become particularly prevalent, ranging from content algorithms shaping political discourse \cite{muscoMinimizingPolarizationDisagreement2018} to personalized dietary suggestions \cite{fontanInfluencingOpinionDynamics2024}. Yet individuals do not passively absorb recommendations but rather process them based on their existing beliefs and preferences. 

Opinion dynamics is an established agent-based framework for modeling how beliefs evolve over social networks, and for studying collective social phenomena emerging from these interactions (for comprehensive tutorials and reviews, see, e.g., \cite{proskurnikovTutorialModelingAnalysis2017,bernardoBoundedConfidenceOpinion2024,tianHowSocialInfluence2023,zinoSocialDiffusionDynamics2023}). 
Recent applications range from food systems \cite{fontanSystemsPerspectivePromoting2025}, politics \cite{fontanSignedNetworkPerspective2021}, and healthcare \cite{anconaModelbasedOpinionDynamics2022}, to the adoption of technologies \cite{aluttoModelPredictiveControl2025}, repeated decision-making \cite{ravazziOptimalPolicyDesign2026}, and recommender systems \cite{lanzettiImpactRecommendationSystems2023}.
Extending the opinion dynamics framework, previous works taking a control perspective on recommendations in social networks use different approaches to model these interactions. Recommendations have been modeled, for instance, as an additional strategic agent \cite{sprengerControlStrategiesRecommendation2024}, by modifying the edges of the underlying graph topology \cite{chenCoevolutionOpinionDynamics2026}, as a committed minority of stubborn agents \cite{raineriControllingSocialNetwork2025}, or through agents' interaction with a recommendation system via clicks \cite{lanzettiImpactRecommendationSystems2023,rossiClosedLoopOpinion2022,chandrasekaranNetworkAwareRecommenderSystem2026}. However, these approaches typically do not model \emph{how} individuals psychologically process the recommendations they receive. This aspect is important to consider, as it affects both the analysis and design of algorithms acting on social networks. In this work, we address this challenge by modeling and analyzing the effect of personal recommendations on opinions in social networks.
The proposed model draws on well-established phenomena from behavioral psychology, in particular \emph{confirmation bias} \cite{nickersonConfirmationBiasUbiquitous1998}, which describes the tendency of individuals to favor opinions close to their own and to actively seek information that supports their existing beliefs. Bounded-confidence models, such as the Hegselmann--Krause model~\cite{rainerOpinionDynamicsBounded2002} (see also \cite{bernardoBoundedConfidenceOpinion2024} for a review), capture this effect by having agents disregard opinions of their neighbors beyond a certain threshold.  
However, in the original formulation by Hegselmann and Krause \cite{rainerOpinionDynamicsBounded2002}, as well as in subsequent extensions, this mechanism is implemented via a discontinuous interaction rule, whose psychological justification has been questioned \cite{masDifferentiationDistancingExplaining2013,kurahashi-nakamuraRobustClusteringGeneralized2016}. 
Instead, more realistic continuous variants of the bounded confidence function, leading to a smooth bounded confidence model, are introduced, for instance, in \cite{deffuantModellingGroupOpinion2004,deffuantHowOpinionsGet2025}. In particular, \cite{deffuantModellingGroupOpinion2004} proposes that the influence of another opinion is weighted by a Gaussian function with standard deviation equal to the uncertainty of the agent.

In this work, we propose a personal recommendation model that extends the DeGroot model of consensus \cite{degrootReachingConsensus1974} to describe the influence of personal recommendation with a Gaussian influence function (similar to the one introduced in \cite{deffuantModellingGroupOpinion2004}), a nonlinear function that is analytically tractable and consistent with the confirmation bias phenomenon described above. 
We analyze the proposed model under different recommendation policies and investigate the resulting collective behavior.
Our main contributions are as follows. First, we consider the baseline case of constant broadcast recommendations towards a target and investigate the qualitative behavior of the system relative to the recommendation value, captured in the model by a positive parameter. This parameter serves as a bifurcation parameter in the analysis; we show that, as the recommendation value increases, the system undergoes a fold bifurcation. For low recommendation values, the system admits a unique consensus equilibrium point.
When the recommendation value exceeds a critical threshold, two additional consensus equilibria emerge, with only the one farther from the recommendation target being locally asymptotically stable, suggesting that more aggressive broadcast recommendations can, in fact, be counterproductive. 
Second, for time-varying targeted recommendations, we show that a greedy policy is optimal in reaching the desired state in shortest time, and we establish convergence to the target opinion from arbitrary initial conditions. Finally, we extend these guarantees to the case where the agents' sensitivity is uncertain by designing a robust policy based on a worst-case estimate.

The remainder of the paper is organized as follows. Section~\ref{sec:model_pbl_formulation} states the model and the problem formulation. In Section~\ref{sec:broadcast} and Section~\ref{sec:greedy}, we investigate broadcast and targeted policies, respectively.
A numerical example is presented in Section~\ref{sec:example}, and conclusions and future directions in Section~\ref{sec:conclusion}.

\textit{Notation:} We denote by $\1$ the vector of all ones of appropriate dimension, $\| \cdot\|_2$ denotes the Euclidean norm. Inequalities between vectors and matrices are understood component-wise. A matrix $W=[w_{ij}] \in \mathbb{R}^{n \times n}$ is row-stochastic if $w_{ij} \ge 0$ for all $i,j$ and $W\1 = \1$. The Perron--Frobenius eigenvalue of a matrix is its spectral radius, which is its largest eigenvalue in absolute value \cite{hornMatrixAnalysis2012}.
A Schur stable matrix has all eigenvalues strictly inside the open unit disk. For a symmetric matrix $W = W^\top$, all eigenvalues are real, and we order them as $\lambda_1(W) \le \cdots \le \lambda_n(W)$.
The undirected graph induced by a symmetric matrix $W$ is denoted $\G(W)=(\nodeSet,\edgeSet,W)$, where $\nodeSet$ is the set of nodes with $\card{\nodeSet}=n$ and $(i,j) \in \edgeSet$ is an edge whenever $w_{ji} > 0$. A graph $\G(W)$ is connected if there exists an edge sequence $(j,i_1),\dots,(i_k,i)$ within $\edgeSet$ between each pair of nodes in $\nodeSet$.

\section{Model and Problem Formulation}\label{sec:model_pbl_formulation}
To model the opinion-forming process in a community of $n$ agents represented by a connected undirected graph $\G(W)$, we consider the following class of nonlinear interconnected systems
\begin{align}
    x(t+1) = aWx(t) + (1-a)\psi(x(t),u(t)).
    \label{eqn:model}
\end{align}
The state vector $x(t) = [x_1(t)\cdots x_n(t)]^T\in \mathcal{X}^n:=[-1,1]^n$ represents the agents' opinions, $u(t)= [u_1(t)\,\cdots\,u_n(t)]^T\in \U^n:=[-1,1]^n$ represents the recommendation, $W = [w_{ij}]\in \mathbb{R}^{n \times n}$ is the adjacency matrix of the network $\G$. The scalar parameter $a\in(0,1)$ represents a trade-off between a social term, describing how agents interact with each other and are influenced by the state of their neighbors, and a term representing the personal recommendation. For $a=1$, \eqref{eqn:model} reduces to the well-known DeGroot model \cite{degrootReachingConsensus1974}.
Each nonlinear function $\psi_i(x_i(t),u_i(t))$ in $\psi(x(t),u(t))=[\psi_1(x_1(t),u_1(t))\,\cdots\,\psi_n(x_n(t),u_n(t))]^T$ describes how the recommended opinion is received by agent $i$. 

In this work, we assume that the matrix $W$ is symmetric, row-stochastic (with a unique Perron-Frobenius eigenvalue equal to $1$). 
We consider the following class of nonlinear functions $\psi_i(x_i,u_i):\, \mathcal{X}\times \U \to \mathcal{X}$, $i=1,\dots,n$,
\begin{align}
    \psi_i(x_i,u_i)= u_i \exp(-\alpha_i (u_i-x_i)^2),\ x_i\in \mathcal{X},u_i\in \U.
\end{align}
Each nonlinear function $\psi_i$ characterizes how agent $i$ accepts the recommendation $u_i$. The choice of $\psi_i$ is motivated by phenomena from psychology describing how individuals perceive recommendations \cite{deffuantModellingGroupOpinion2004}. In particular, it captures the idea that $u_i = x_i$ is a neutral recommendation that makes agent $i$ keep its current belief, whereas when $\abs{u_i - x_i} \gg 0$ agent $i$ tends to reject the recommendation $u_i$. Motivated by the confirmation bias effect from psychology discussed in Section~\ref{sec:intro}, the function $\psi_i$ decreases smoothly with increasing distance between the recommendation $u_i$ and the actual opinion $x_i$. The parameter $\alpha_i$ represents agent $i$'s sensitivity to a recommendation, i.e., a larger $\alpha_i$ leads to a smaller acceptance window as the influence function decays faster. In this preliminary work, we make the following assumption: 
\begin{assumption}
   We assume the homogeneous case \mbox{$\alpha_1=\dots=\alpha_n=:\alpha>0$} for the agents' sensitivities.\label{ass:alpha_homogeneous}
\end{assumption}
The following lemma shows that model~\eqref{eqn:model} is well-defined.
\begin{lemma}
    The model~\eqref{eqn:model} is well-posed, i.e., $x(t)\in \mathcal{X}^n$ for all $t$ given that $u(t)\in \U^n$ for all $t$ and $x(0)\in \mathcal{X}^n$.
\end{lemma}
\begin{proof}
    Suppose $x(t)\in \mathcal{X}^n$. From $W$ row-stochastic, it follows that $-\1\leq W x(t)\leq \1$. Furthermore, $0\leq\exp(-\alpha(x_i-u_i)^2)\leq 1$ so $\lvert \psi_i\rvert\leq \lvert u_i\rvert\leq 1$, thus $-\1\leq\psi(x(t),u(t))\leq \1$. Combining with $a\in(0,1)$, we obtain $x(t+1)\in \mathcal{X}^n$, and the result follows by induction.
\end{proof}

Given the model \eqref{eqn:model} and Assumption~\ref{ass:alpha_homogeneous}, we want to investigate: (i) whether a consensus to a desired target opinion $\xtarget=\1$ can be reached with a time-invariant broadcast recommendation $u(t)=u\1$ for all $t$; (ii) what the optimal time-varying and targeted recommendations $u_i(t)$, $i=1,\dots,n$, are that guarantee convergence to $\xtarget= \1$; and (iii) whether robust convergence can be guaranteed for an uncertain sensitivity $\alpha \in [\alpha^-, \alpha^+]$.

\begin{remark}
    In practice, it is common to assume a target $\xtarget\in\{0,\1\}$ \cite{fontanInfluencingOpinionDynamics2024,ravazziOptimalPolicyDesign2026}. 
    This also justifies the restriction to positive recommendations $u_i>0$ in the analysis of Section~\ref{sec:broadcast}.
\end{remark}

\section{Time-invariant broadcast recommendations}\label{sec:broadcast}
We first study the case of constant recommendations. 
In particular, we assume that $u(t) = u\1$, where $u$ is a scalar, constant across agents and time, representing a time-invariant broadcast strategy. The goal of the centralized recommendations is to guide the collective opinion towards the target $\xtarget=\1$. It is then natural to interpret $u$ as a recommendation towards that target and to assume $u \in \U_+ \coloneqq [0,1]$.
Our goal is to study the qualitative behavior of the nonlinear system~\eqref{eqn:model} as the recommendation $u$ varies, effectively performing a bifurcation analysis with $u$ as the bifurcation parameter. We investigate the existence and stability of the equilibrium points of \eqref{eqn:model} (Theorems~\ref{thm:existence} and \ref{thm:stability}, respectively), and provide an interpretation of these results in the context of social recommendations.

The following theorem establishes the existence of \textit{consensus equilibrium points} for the system~\eqref{eqn:model}, defined as equilibrium points in $\vspan{\1}$. In particular, it shows that there exists a threshold value for $u$ above which the system~\eqref{eqn:model} exhibits a fold bifurcation (also called saddle-node or tangent bifurcation) \cite{kuznetsovElementsAppliedBifurcation2023}.
\begin{theorem}[Existence]\label{thm:existence}
    Assume $u$ is constant and homogeneous across agents, i.e., $u(t) = u \1$ for all $t$, $u\in \U_+$.
    
    \begin{enumerate}[label=(\roman*)]
    \item For any value of $u\in \U_+$, the system~\eqref{eqn:model} admits a consensus equilibrium point given by
    \begin{equation}
        \xeq = u \1.\label{eqn:consensus_eq}
    \end{equation}
    \item Assume that $\alpha>2$. There exists a threshold value $\uast$ for $u\in \U_+$ such that: If $u<\uast$, $\xeq$ in \eqref{eqn:consensus_eq} is the unique consensus equilibrium point for the system~\eqref{eqn:model}.
    When $u$ crosses $\uast$, the system undergoes a fold bifurcation and two new consensus equilibria appear.
    The threshold value $\uast$ solves the nonlinear equation 
    \begin{align}
    \uast \in (\sqrt{2/\alpha},1]:\quad \xi(\uast) &= \uast e^{-\alpha (\uast-\xi(\uast))^2},
    \label{eqn:u_ast}
    \end{align}
    where $\xi(\uast) = \frac{\uast}{2} - \sqrt{\big(\frac{\uast}{2}\big)^2 - \frac{1}{2\alpha}}$.
    \end{enumerate}
    \vspace{0.1cm}
\end{theorem}
\begin{proof}
First, it is possible to observe that the consensus space is invariant. Let $\Pi = \frac{1}{n}\1 \1^T$ and $I-\Pi$ be the projection operators onto the consensus space $\vspan{\1}$ and onto the orthogonal space $\vspan{\1^\perp}$, respectively.
We can decompose any vector $x \in \mathcal{X}^n$ as $x= \xi \1 + z$, where $\xi = \frac{1}{n} \1^T x\in \mathcal{X}$ and $z \in \vspan{\1^\perp}\cap\mathcal{X}^n$ (i.e., $\1^T z=0$).
Since $W$ is row-stochastic and $W=W^T$, the projection of the system~\eqref{eqn:model} in the coordinates $(\xi_t,z_t):=(\xi(t),z(t))$ can be rewritten as
\begin{subequations}
\begin{align}
    \xi_{t+1} &= a \xi_{t} + (1-a) \frac{1}{n}\1^T \psi( \xi_{t}\1 + z_{t}, u)
    \label{eqn:proj_model_consensus_0}
    \\
    z_{t+1} &= a W z_{t} + (1-a) (I-\Pi) \psi( \xi_{t}\1 + z_{t}, u)
\end{align}   
\label{eqn:proj_model}
\end{subequations}
Since the nonlinearities are identical under the assumption of homogeneous $\alpha$ and $u$ (i.e., $\psi_1=\dots = \psi_n$, see Appendix~\ref{sec:propertiesNonlinearFunctions}), $(I-\Pi) \psi( \xi_{t}\1, u)=0$, that is, the consensus space is invariant.
Therefore, \eqref{eqn:proj_model_consensus_0} reduces to
\begin{align}
    \xi_{t+1} &= a \xi_{t} + (1-a) \psi(\xi_{t}, u),
    \label{eqn:proj_model_consensus}
\end{align}
where, with some abuse of notation, we use $\psi( \xi_{t} \1, u) = \psi( \xi_{t}, u) \1$.
The equilibrium points $\xi^\ast$ of \eqref{eqn:proj_model_consensus} correspond to the consensus equilibrium points of the system~\eqref{eqn:model}; moreover, from \eqref{eqn:proj_model_consensus}, $x = \xi\1 \in \mathcal{X}^n$ is a consensus equilibrium point for the system~\eqref{eqn:model} if and only if $ \xi = \psi(\xi,u)$.
\smallskip

(i) We note that $\xi = u$ is always a solution of $\xi = \psi(\xi,u)$, meaning that $\xeq = u \1$ is an admissible equilibrium point of the system~\eqref{eqn:model} for any value of $u \in \U_+$.
\smallskip

(ii) We now prove that the system~\eqref{eqn:proj_model_consensus} undergoes a fold bifurcation at $u=\uast$, where $\uast$ is defined in \eqref{eqn:u_ast}. The proof follows \cite[Ch.~4--5]{kuznetsovElementsAppliedBifurcation2023}.

The 1-dimensional system~\eqref{eqn:proj_model_consensus} undergoes a fold bifurcation at $(\xi^\ast,\uast)$ if the following system of equations holds:
\begin{align}
\begin{cases}
    \xi^\ast = \psi(\xi^\ast,\uast)\\
    1 = \pde{\psi}{\xi}(\xi^\ast,\uast)\\
    \ppde{\psi}{\xi}(\xi^\ast,\uast)\ne 0 \\
    \pde{\psi}{u}(\xi^\ast,\uast)\ne 0\\
\end{cases}\!\!\!\!\!\!
\Leftrightarrow
\begin{cases}
    \xi^\ast = \uast e^{-\alpha (\uast-\xi^\ast)^2}
    \\
    1 = 2 \alpha \uast (\uast-\xi^\ast)  e^{-\alpha (\uast-\xi^\ast)^2}
    \\
    \pde{\psi}{\xi}(\xi^\ast\!,\uast) (\uast\!-\!\xi^\ast)\! -\! \psi(\xi^\ast\!,\uast)\! \ne\! 0
    \\
    1 - 2\alpha (\uast-\xi^\ast) \uast \ne 0 
\end{cases}
\label{eqn:admissible_eqpoint}
\end{align}
The conditions in \eqref{eqn:admissible_eqpoint} are, respectively, the fixed-point, fold bifurcation, nondegeneracy, and transversality conditions \cite[Theorem 4.2]{kuznetsovElementsAppliedBifurcation2023}. 
Given $\xi^\ast = \uast e^{-\alpha (\uast-\xi^\ast)^2}$, the fold bifurcation condition implies\vspace{-.3em}
\begin{align*}
    1 = 2 \alpha (\uast-\xi^\ast) \xi^\ast
    &\;\;\Leftrightarrow\;\; 2 \alpha (\xi^\ast)^2 -2 \alpha \uast \xi^\ast + 1 = 0   
\\
    \Leftrightarrow\;\; \xi_{1,2}^\ast(\uast) &= 
    \frac{\uast}{2}\pm \sqrt{ \Big(\frac{\uast}{2}\Big)^2 - \frac{1}{2\alpha} }.
\end{align*}
The discriminant exists when $\uast > \sqrt{2/\alpha}$. 
Among the two possible choices, namely, $\xi_{1}^\ast(\uast)$ ($-$) and $\xi_{2}^\ast(\uast)$ ($+$), $\xi_{2}^\ast(\uast)$ is not admissible since it does not satisfy both the fixed-point and fold bifurcation equations in \eqref{eqn:admissible_eqpoint}. One quick way to verify this is by observing that $\xi_{2}^\ast(\uast) $ is located to the right of the first inflection point, i.e., $\xi_{2}^\ast(\uast)> \uast-\frac{1}{\sqrt{2\alpha}}$. Indeed \vspace{-.3em}
\begin{align*}
    \frac{\uast}{2}&+\sqrt{ \Big(\frac{\uast}{2}\Big)^2 - \frac{1}{2\alpha} } > \uast-\frac{1}{\sqrt{2\alpha}}
    \\
    \Leftrightarrow\quad&
    \sqrt{ \Big(\frac{\uast}{2}\Big)^2 - \frac{1}{2\alpha} } > \frac{\uast}{2}-\frac{1}{\sqrt{2\alpha}}
    \\
    \Leftrightarrow\quad&
    \Big(\frac{\uast}{2}\Big)^2 - \frac{1}{2\alpha} > \Big(\frac{\uast}{2}\Big)^2 +\frac{1}{2\alpha}-\frac{u}{\sqrt{2\alpha}}
    \\
    \Leftrightarrow\quad&
    u > \sqrt{2/\alpha},
\end{align*}
which is satisfied as it represents the condition for the existence of the discriminant.
Therefore, the unique admissible point satisfying the fixed-point condition and the fold bifurcation condition in \eqref{eqn:admissible_eqpoint} is given by
$\xi_{1}^\ast(\uast)= \frac{\uast}{2}- \sqrt{ \big(\frac{\uast}{2}\big)^2 - \frac{1}{2\alpha} }$. Consequently, the threshold value $\uast$ solves the nonlinear equation 
\begin{align*}
\xi_1^\ast(\uast) = \uast e^{-\alpha (\uast-\xi_1^\ast(\uast))^2},\quad \uast>\sqrt{2/\alpha},
\end{align*}
for which a closed-form may not be derived. Finally, observe that $\pde{\psi}{\xi}(\xi_1^\ast(\uast),\uast)(\uast-\xi_1^\ast(\uast))-\psi(\xi_1^\ast(\uast),\uast) = \uast-2\xi_1^\ast(\uast) >0$, and $2\alpha (\uast-\xi_1^\ast(\uast)) \uast = \uast/\xi_1^\ast(\uast) > 1$. Therefore, the nondegeneracy conditions and transversality conditions in \eqref{eqn:admissible_eqpoint} are also satisfied, which completes the recognition problem for a fold bifurcation.
\end{proof}

A consequence of Theorem~\ref{thm:existence} is that the system~\eqref{eqn:model} cannot admit more than three consensus equilibrium points; moreover, it follows directly from the proof of Theorem~\ref{thm:existence} that, component-wise, each consensus equilibrium point lies in $[0,u]$.
We conjecture that system~\eqref{eqn:model} cannot admit any other equilibria, that is, disagreement equilibrium points of the form $x = \xi \1 + z$, with $z\ne 0$ and $z\in \vspan{\1^\perp}\cap\mathcal{X}^n$, cannot exist. We leave the proof of the conjecture for future work. Nonetheless, it can be shown that, component-wise, each equilibrium point is upper bounded by $u$ (Lemma~\ref{lem:location_eq}). 
\begin{conjecture}
Let $L_{\psi,\max}:=\displaystyle\max_{u\in\U_+} L_{\psi}(u) = \sqrt{2\alpha} e^{-1/2}$, where $L_{\psi}(u)$ is the Lipschitz constant of $\psi$. 
If $a>a_{\min}: = \displaystyle\frac{L_{\psi,\max}-1}{L_{\psi,\max}-\lambda_{n-1}(W)}$ and $u\in\U_+$, then the system~\eqref{eqn:model} admits only consensus equilibrium points.
\end{conjecture}
\begin{lemma}\label{lem:location_eq}
    Assume $u\in\U_+$. If $\xeq \in \mathcal{X}^n$ is an equilibrium point of \eqref{eqn:model}, then $\xeq_i \le u$ for all $i=1,\dots,n$.
\end{lemma}
\begin{proof}
    Assume, by contradiction, that there exists an index $i$ s.t. $\xeq_i > u$. Let $x_M := \max_{i=1,\dots,n} \xeq_i>u$. By the equilibrium condition, we have\vspace{-0.1cm}
    \begin{align*}
        \xeq_M &= a \scalebox{0.95}{$\sum$}_{j=1}^n w_{Mj} \xeq_j + (1-a) \psi_M(\xeq_M,u)
        \\&\le a \scalebox{0.95}{$\sum$}_{j=1}^n w_{Mj} \xeq_M + (1-a) \psi_M(\xeq_M,u)
        \\&= a \xeq_M + (1-a) \psi_M(\xeq_M,u)
    \end{align*}
    Since $\psi_M(\xeq_M,u) < \xeq_M$ for $\xeq_M > u$, it follows that $\xeq_M < (a + 1-a) \xeq_M = \xeq_M$, which yields a contradiction.
\end{proof}
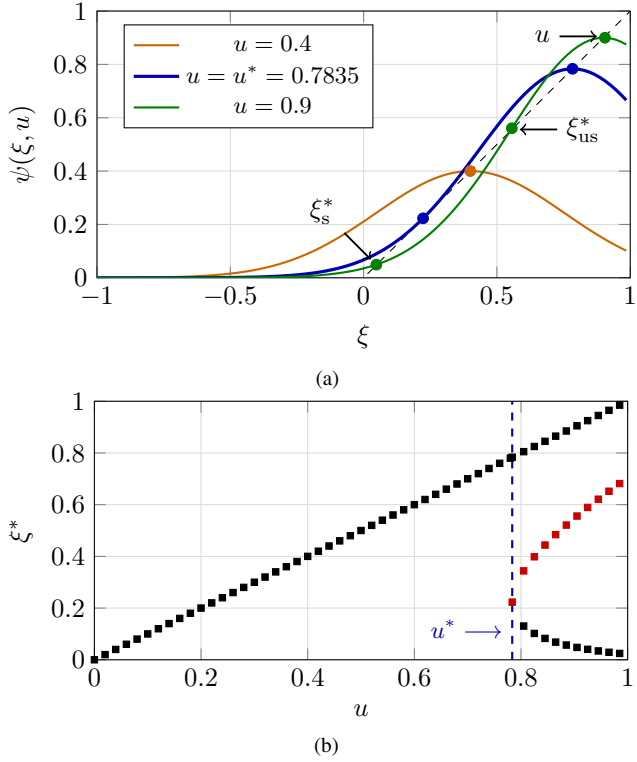
\begin{figure}[t!]\centering
    \subfloat[]{\begin{tikzpicture}
\begin{axis}[
  xlabel={$\xi$}, ylabel={$\psi(\xi,u)$},
  xmin=-1, xmax=1, ymin=0, ymax=1,
  axis equal image,
  grid=major, grid style={line width=0.3pt, draw=gray!30},
  width=\linewidth, height=5.5cm,
  legend style={at={(0.05,0.95)}, anchor=north west, font=\small},
]
\addplot[forget plot, dashed, black, domain=-1:1, samples=2] {x};
\addplot[orange!80!black, thick] coordinates {
  (-1.0000,0.0002)
  (-0.9799,0.0002)
  (-0.9599,0.0002)
  (-0.9398,0.0003)
  (-0.9198,0.0004)
  (-0.8997,0.0005)
  (-0.8797,0.0006)
  (-0.8596,0.0007)
  (-0.8396,0.0009)
  (-0.8195,0.0010)
  (-0.7995,0.0013)
  (-0.7794,0.0015)
  (-0.7594,0.0018)
  (-0.7393,0.0022)
  (-0.7193,0.0027)
  (-0.6992,0.0032)
  (-0.6792,0.0038)
  (-0.6591,0.0045)
  (-0.6391,0.0053)
  (-0.6190,0.0063)
  (-0.5990,0.0074)
  (-0.5789,0.0087)
  (-0.5589,0.0101)
  (-0.5388,0.0118)
  (-0.5188,0.0137)
  (-0.4987,0.0158)
  (-0.4787,0.0182)
  (-0.4586,0.0210)
  (-0.4386,0.0240)
  (-0.4185,0.0274)
  (-0.3985,0.0312)
  (-0.3784,0.0354)
  (-0.3584,0.0401)
  (-0.3383,0.0452)
  (-0.3183,0.0508)
  (-0.2982,0.0569)
  (-0.2782,0.0635)
  (-0.2581,0.0707)
  (-0.2381,0.0785)
  (-0.2180,0.0868)
  (-0.1980,0.0957)
  (-0.1779,0.1051)
  (-0.1579,0.1152)
  (-0.1378,0.1258)
  (-0.1178,0.1369)
  (-0.0977,0.1485)
  (-0.0777,0.1606)
  (-0.0576,0.1731)
  (-0.0376,0.1860)
  (-0.0175,0.1992)
  (0.0025,0.2126)
  (0.0226,0.2262)
  (0.0426,0.2400)
  (0.0627,0.2537)
  (0.0827,0.2674)
  (0.1028,0.2809)
  (0.1228,0.2942)
  (0.1429,0.3070)
  (0.1629,0.3195)
  (0.1830,0.3313)
  (0.2030,0.3425)
  (0.2231,0.3529)
  (0.2431,0.3625)
  (0.2632,0.3711)
  (0.2832,0.3788)
  (0.3033,0.3853)
  (0.3233,0.3907)
  (0.3434,0.3949)
  (0.3634,0.3979)
  (0.3835,0.3996)
  (0.4035,0.4000)
  (0.4236,0.3991)
  (0.4436,0.3970)
  (0.4637,0.3936)
  (0.4837,0.3889)
  (0.5038,0.3831)
  (0.5238,0.3762)
  (0.5439,0.3682)
  (0.5639,0.3592)
  (0.5840,0.3494)
  (0.6040,0.3387)
  (0.6241,0.3272)
  (0.6441,0.3152)
  (0.6642,0.3026)
  (0.6842,0.2896)
  (0.7043,0.2762)
  (0.7243,0.2626)
  (0.7444,0.2489)
  (0.7644,0.2352)
  (0.7845,0.2215)
  (0.8045,0.2079)
  (0.8246,0.1945)
  (0.8446,0.1814)
  (0.8647,0.1686)
  (0.8847,0.1563)
  (0.9048,0.1444)
  (0.9248,0.1329)
  (0.9449,0.1220)
  (0.9649,0.1116)
  (0.9850,0.1018)
};
\addlegendentry{$u = 0.4$}
\addplot[forget plot, only marks, mark=*, mark size=2pt, mark options={fill=orange!80!black, draw=orange!80!black}] coordinates {
  (0.4000,0.4000)
};
\addplot[blue!70!black, very thick] coordinates {
  (-1.0000,0.0000)
  (-0.9799,0.0000)
  (-0.9599,0.0000)
  (-0.9398,0.0000)
  (-0.9198,0.0000)
  (-0.8997,0.0000)
  (-0.8797,0.0000)
  (-0.8596,0.0000)
  (-0.8396,0.0000)
  (-0.8195,0.0000)
  (-0.7995,0.0000)
  (-0.7794,0.0000)
  (-0.7594,0.0001)
  (-0.7393,0.0001)
  (-0.7193,0.0001)
  (-0.6992,0.0001)
  (-0.6792,0.0002)
  (-0.6591,0.0002)
  (-0.6391,0.0002)
  (-0.6190,0.0003)
  (-0.5990,0.0004)
  (-0.5789,0.0005)
  (-0.5589,0.0006)
  (-0.5388,0.0007)
  (-0.5188,0.0009)
  (-0.4987,0.0011)
  (-0.4787,0.0013)
  (-0.4586,0.0016)
  (-0.4386,0.0020)
  (-0.4185,0.0024)
  (-0.3985,0.0029)
  (-0.3784,0.0035)
  (-0.3584,0.0043)
  (-0.3383,0.0051)
  (-0.3183,0.0061)
  (-0.2982,0.0073)
  (-0.2782,0.0086)
  (-0.2581,0.0102)
  (-0.2381,0.0121)
  (-0.2180,0.0142)
  (-0.1980,0.0166)
  (-0.1779,0.0194)
  (-0.1579,0.0226)
  (-0.1378,0.0263)
  (-0.1178,0.0304)
  (-0.0977,0.0351)
  (-0.0777,0.0403)
  (-0.0576,0.0462)
  (-0.0376,0.0528)
  (-0.0175,0.0602)
  (0.0025,0.0683)
  (0.0226,0.0773)
  (0.0426,0.0872)
  (0.0627,0.0980)
  (0.0827,0.1099)
  (0.1028,0.1228)
  (0.1228,0.1367)
  (0.1429,0.1517)
  (0.1629,0.1679)
  (0.1830,0.1852)
  (0.2030,0.2036)
  (0.2231,0.2231)
  (0.2431,0.2436)
  (0.2632,0.2653)
  (0.2832,0.2879)
  (0.3033,0.3115)
  (0.3233,0.3359)
  (0.3434,0.3610)
  (0.3634,0.3868)
  (0.3835,0.4131)
  (0.4035,0.4398)
  (0.4236,0.4666)
  (0.4436,0.4936)
  (0.4637,0.5204)
  (0.4837,0.5469)
  (0.5038,0.5729)
  (0.5238,0.5983)
  (0.5439,0.6227)
  (0.5639,0.6461)
  (0.5840,0.6681)
  (0.6040,0.6888)
  (0.6241,0.7077)
  (0.6441,0.7249)
  (0.6642,0.7401)
  (0.6842,0.7532)
  (0.7043,0.7641)
  (0.7243,0.7726)
  (0.7444,0.7787)
  (0.7644,0.7823)
  (0.7845,0.7835)
  (0.8045,0.7821)
  (0.8246,0.7782)
  (0.8446,0.7719)
  (0.8647,0.7631)
  (0.8847,0.7520)
  (0.9048,0.7387)
  (0.9248,0.7233)
  (0.9449,0.7060)
  (0.9649,0.6868)
  (0.9850,0.6661)
};
\addlegendentry{$u = u^* = 0.7835$}
\addplot[forget plot, only marks, mark=*, mark size=2pt, mark options={fill=blue!70!black, draw=blue!70!black}] coordinates {
  (0.7835,0.7835)
  (0.2230,0.2230)
};
\addplot[green!50!black, thick] coordinates {
  (-1.0000,0.0000)
  (-0.9799,0.0000)
  (-0.9599,0.0000)
  (-0.9398,0.0000)
  (-0.9198,0.0000)
  (-0.8997,0.0000)
  (-0.8797,0.0000)
  (-0.8596,0.0000)
  (-0.8396,0.0000)
  (-0.8195,0.0000)
  (-0.7995,0.0000)
  (-0.7794,0.0000)
  (-0.7594,0.0000)
  (-0.7393,0.0000)
  (-0.7193,0.0000)
  (-0.6992,0.0000)
  (-0.6792,0.0000)
  (-0.6591,0.0001)
  (-0.6391,0.0001)
  (-0.6190,0.0001)
  (-0.5990,0.0001)
  (-0.5789,0.0001)
  (-0.5589,0.0002)
  (-0.5388,0.0002)
  (-0.5188,0.0003)
  (-0.4987,0.0004)
  (-0.4787,0.0004)
  (-0.4586,0.0006)
  (-0.4386,0.0007)
  (-0.4185,0.0009)
  (-0.3985,0.0011)
  (-0.3784,0.0013)
  (-0.3584,0.0016)
  (-0.3383,0.0020)
  (-0.3183,0.0024)
  (-0.2982,0.0029)
  (-0.2782,0.0035)
  (-0.2581,0.0042)
  (-0.2381,0.0051)
  (-0.2180,0.0061)
  (-0.1980,0.0072)
  (-0.1779,0.0086)
  (-0.1579,0.0102)
  (-0.1378,0.0121)
  (-0.1178,0.0143)
  (-0.0977,0.0168)
  (-0.0777,0.0197)
  (-0.0576,0.0230)
  (-0.0376,0.0267)
  (-0.0175,0.0310)
  (0.0025,0.0359)
  (0.0226,0.0414)
  (0.0426,0.0476)
  (0.0627,0.0545)
  (0.0827,0.0622)
  (0.1028,0.0708)
  (0.1228,0.0803)
  (0.1429,0.0909)
  (0.1629,0.1024)
  (0.1830,0.1151)
  (0.2030,0.1289)
  (0.2231,0.1439)
  (0.2431,0.1602)
  (0.2632,0.1777)
  (0.2832,0.1965)
  (0.3033,0.2166)
  (0.3233,0.2380)
  (0.3434,0.2606)
  (0.3634,0.2845)
  (0.3835,0.3096)
  (0.4035,0.3358)
  (0.4236,0.3630)
  (0.4436,0.3912)
  (0.4637,0.4202)
  (0.4837,0.4500)
  (0.5038,0.4803)
  (0.5238,0.5110)
  (0.5439,0.5419)
  (0.5639,0.5728)
  (0.5840,0.6036)
  (0.6040,0.6339)
  (0.6241,0.6637)
  (0.6441,0.6926)
  (0.6642,0.7205)
  (0.6842,0.7471)
  (0.7043,0.7721)
  (0.7243,0.7955)
  (0.7444,0.8169)
  (0.7644,0.8362)
  (0.7845,0.8532)
  (0.8045,0.8678)
  (0.8246,0.8797)
  (0.8446,0.8890)
  (0.8647,0.8955)
  (0.8847,0.8992)
  (0.9048,0.8999)
  (0.9248,0.8978)
  (0.9449,0.8928)
  (0.9649,0.8850)
  (0.9850,0.8744)
};
\addlegendentry{$u = 0.9$}
\addplot[forget plot, only marks, mark=*, mark size=2pt, mark options={fill=green!50!black, draw=green!50!black}] coordinates {
  (0.0479,0.0493)
  (0.9050,0.8999)
  (0.5560,0.5607)
};
\node[coordinate, pin={[pin edge={<-, black, line width=0.7pt, shorten <=3pt}, pin distance=0.6cm]135:$\xi^{\ast}_\mathrm{s}$}] at (axis cs:0.0479,0.0479) {};
\node[coordinate, pin={[pin edge={<-, black, line width=0.7pt, shorten <=3pt}, pin distance=0.6cm]0:$\xi^{\ast}_\mathrm{us}$}] at (axis cs:0.5560,0.5560) {};
\node[coordinate, pin={[pin edge={<-, black, line width=0.7pt, shorten <=3pt}, pin distance=0.6cm]180:$u$}] at (axis cs:0.9050,0.9050) {};
\end{axis}
\end{tikzpicture}}\\[-3pt]
    \subfloat[]{\begin{tikzpicture}
\begin{axis}[
  xlabel={$u$}, ylabel={$\xi^*$},
  xmin=0, xmax=1, ymin=0, ymax=1,
  grid=major, grid style={line width=0.3pt, draw=gray!30},
  width=\linewidth, height=5cm, clip=false,
]
\addplot[only marks, mark=square*, mark size=1.2pt, mark options={fill=black, draw=black}] coordinates {
  (0.0000,0.0000)
  (0.0200,0.0200)
  (0.0400,0.0400)
  (0.0600,0.0600)
  (0.0800,0.0800)
  (0.1000,0.1000)
  (0.1200,0.1200)
  (0.1400,0.1400)
  (0.1600,0.1600)
  (0.1800,0.1800)
  (0.2000,0.2000)
  (0.2200,0.2200)
  (0.2400,0.2400)
  (0.2600,0.2600)
  (0.2800,0.2800)
  (0.3000,0.3000)
  (0.3200,0.3200)
  (0.3400,0.3400)
  (0.3600,0.3600)
  (0.3800,0.3800)
  (0.4000,0.4000)
  (0.4200,0.4200)
  (0.4400,0.4400)
  (0.4600,0.4600)
  (0.4800,0.4800)
  (0.5000,0.5000)
  (0.5200,0.5200)
  (0.5400,0.5400)
  (0.5600,0.5600)
  (0.5800,0.5800)
  (0.6000,0.6000)
  (0.6200,0.6200)
  (0.6400,0.6400)
  (0.6600,0.6600)
  (0.6800,0.6800)
  (0.7000,0.7000)
  (0.7200,0.7200)
  (0.7400,0.7400)
  (0.7600,0.7600)
  (0.7800,0.7800)
  (0.7835,0.7835)
  (0.8050,0.8050)
  (0.8050,0.1304)
  (0.8250,0.1019)
  (0.8250,0.8250)
  (0.8450,0.0827)
  (0.8450,0.8450)
  (0.8650,0.0683)
  (0.8650,0.8650)
  (0.8850,0.0570)
  (0.8850,0.8850)
  (0.9050,0.0479)
  (0.9050,0.9050)
  (0.9250,0.0405)
  (0.9250,0.9250)
  (0.9450,0.0342)
  (0.9450,0.9450)
  (0.9650,0.0290)
  (0.9650,0.9650)
  (0.9850,0.9850)
  (0.9850,0.0246)
};
\addplot[only marks, mark=square*, mark size=1.2pt, mark options={fill=red!80!black, draw=red!80!black}] coordinates {
  (0.7835,0.2230)
  (0.8050,0.3439)
  (0.8250,0.3984)
  (0.8450,0.4437)
  (0.8650,0.4840)
  (0.8850,0.5211)
  (0.9050,0.5560)
  (0.9250,0.5892)
  (0.9450,0.6211)
  (0.9650,0.6518)
  (0.9850,0.6816)
};
\draw[dashed, blue!70!black, line width=0.8pt] (axis cs:0.7835,0) -- (axis cs:0.7835,1);
\node[blue!70!black, anchor=east, font=\small] at (axis cs:0.7835,0.12) {$u^* \longrightarrow$};
\end{axis}
\end{tikzpicture}}
    \caption{Bifurcation analysis of the system~\eqref{eqn:model} with a constant (and homogeneous across agents) $u$ and $\alpha=4$. The threshold value $\uast$ as defined in Theorem~\ref{thm:existence} is given by $\uast = 0.7835$.
    (a): Illustration of the nonlinearity projected on the consensus space, $\psi(\xi,u)$, with $\xi \in \mathcal{X}$, for three values of $u$: $u=0.4$, $u=\uast$, and $u=0.9$. The dashed black line depicts the identity map $\xi \mapsto \xi$. When $u<\uast$, $\psi(\xi,u)$ and the identity intersect at a single point (orange), yielding a unique consensus equilibrium for the system~\eqref{eqn:model}. At the threshold value $u=\uast$, the identity map is tangent to $\psi(\xi,u)$ (blue), indicating a fold bifurcation. For $u>\uast$, the two curves intersect at three distinct points (green), corresponding to three consensus equilibrium states ($\xeq=u\1$, $\xeq_{\mathrm{s}} = \xi_\mathrm{s}^\ast \1$, and $\xeq_{\mathrm{us}} = \xi_\mathrm{us}^\ast \1$).
    (b): Bifurcation diagram of \eqref{eqn:model} in the consensus space. Legend: x-axis: bifurcation parameter $u \in \{0,0.01, \dots, \uast, \dots, 1\}$. y-axis: Equilibrium points $\xi^\ast$ of the system \eqref{eqn:proj_model_consensus} for each value of $u$. Colormap: red = unstable, black = locally asymptotically stable.}\vspace{-0.5cm}
    \label{fig:bifurcation}
\end{figure}

The following theorem studies the stability of consensus equilibrium points. In particular, 
Theorem~\ref{thm:stability}\,(i) introduces a sufficient (and conservative) condition for the system~\eqref{eqn:model} to admit a unique and globally asymptotically stable equilibrium point. Theorem~\ref{thm:stability}\,(ii) presents local asymptotic stability condition for all the consensus equilibrium points of the system~\eqref{eqn:model}, whose existence was established in Theorem~\ref{thm:existence}.

\begin{theorem}[Stability]\label{thm:stability}
Assume $u\in\U_+$ is constant and homogeneous across agents, and $\alpha>2$. Let $\uast$ be defined as in \eqref{eqn:u_ast}, $\xeq$ as in \eqref{eqn:consensus_eq}, and $\xeq_\mathrm{s},\xeq_\mathrm{us}$ be the consensus equilibrium points of the system~\eqref{eqn:model} which arise at the fold bifurcation.
\begin{enumerate}[label=(\roman*)]
    \item If $u< \tilde{u}^\ast:= \frac{e^{1/2}}{\sqrt{2\alpha}}< \uast$, then the equilibrium point $\xeq = u \1$ of \eqref{eqn:model} is unique and globally exponentially stable. 
    \item If $u>\uast$, then the equilibrium points $\xeq = u \1$ and $\xeq_\mathrm{s}$ are locally asymptotically stable, while $\xeq_\mathrm{us} $ is unstable.
\end{enumerate}
\end{theorem}
\begin{proof}
(i) First, we show that $\tilde{u}^\ast < \uast$. This follows trivially from the fact that $\tilde{u}^\ast < \sqrt{2/\alpha}$ and that $\sqrt{2/\alpha} < \uast$.

Let $\Phi(x(t))= a W x(t) + (1-a)\psi(x(t),u)$, and observe that the equilibrium points of the system~\eqref{eqn:model} correspond to the fixed-points of $\Phi(x)$.
For a fixed constant value of $u$, the Jacobian of the system \eqref{eqn:model} is given by $J(x):=a W + (1-a) \Psi_x(x)$, where $\Psi_x(x): = \diag{\pde{\psi_1}{x_1}(x_1,u),\dots,\pde{\psi_n}{x_n}(x_n,u)}$.
Since $W$ is symmetric and $\Psi_x(x)$ is a diagonal matrix, then $J(x)$ is also symmetric with real  eigenvalues; moreover, by Weyl's theorem \cite[Thm 4.3.1]{hornMatrixAnalysis2012}, it holds that $\lambda_n(J(x)) \le a \lambda_n(W)+(1-a)\lambda_n(\Psi_x(x)) = a + (1-a) \max_{i} \pde{\psi_i}{x_i}(x_i,u) \le a + (1-a)L_\psi(u)$, where $L_\psi(u)$ is the Lipschitz constant of the nonlinear function $\psi$ (see Appendix~\ref{sec:propertiesNonlinearFunctions}).
Therefore, if $u$ is s.t. $a + (1-a)L_\psi(u)<1$ then, according to the Banach contraction theorem for discrete-time dynamics (see e.g. \cite[Ch.~3.4]{bulloContractionTheoryDynamical2026}), $\Phi$ is contracting and $\Phi$ has a unique, globally exponentially equilibrium point $\xeq$, which, as per Theorem~\ref{thm:existence}, must be $\xeq = u\1$. Since the Lipschitz constant of $\psi$ is given by $L_{\psi}(u)= \sqrt{2\alpha} u e^{-1/2}$ (see Appendix~\ref{sec:propertiesNonlinearFunctions}), then $\Phi$ is contracting iff $a + (1-a)\sqrt{2\alpha} u e^{-1/2} < 1$, that is, $u < \tilde{u}^\ast$.

(ii) Assume $u>\uast$. 
To analyze the local stability of $\xeq = u \1$, it is sufficient to observe that the Jacobian $J(\xeq)= a W$ is a Schur matrix for all $u\in \U_+$. Thus, $\xeq$ is locally asymptotically stable.

Studying the local stability of the equilibria $\xeq_{\mathrm{s}}, \xeq_{\mathrm{us}} \in \vspan{\1}$  reduces to examining the eigenvalues of the Jacobian $J(x) = a W + (1-a) \Psi_x(x)$. Since $W$ is symmetric and $\Psi_x(x)$ is diagonal, $J(x)$ is symmetric and therefore has real eigenvalues.
Moreover, in the consensus space all derivatives of the nonlinearities coincide; that is, $\Psi_x(x) = \pde{\psi}{\xi}(\xi,u) I$ when $x = \xi \1$ (with some abuse of notation). Hence, it holds that $\lambda_n(J(x)) = a \lambda_n(W)+(1-a)\pde{\psi}{\xi}(\xi,u) = a + (1-a) \pde{\psi}{\xi}(\xi,u)$. 
Let $\xeq_{\mathrm{s}} = \xi_\mathrm{s}^\ast \1$ and $\xeq_{\mathrm{us}} = \xi_\mathrm{us}^\ast \1$. By construction, $\pde{\psi}{\xi}(\xi_{\mathrm{s}}^\ast,u)<1$ and $\pde{\psi}{\xi}(\xi_{\mathrm{us}}^\ast,u)>1$.
Therefore, the largest eigenvalues of the Jacobian of the system~\eqref{eqn:model} at these equilibria satisfy:\vspace{-0.1cm}
\begin{align*}
     \lambda_n(J(\xeq_{\mathrm{s}}))&= a + (1-a) \pde{\psi}{\xi}(\xi_{\mathrm{s}}^\ast,u) < a + 1-a = 1
     \\
     \lambda_n(J(\xeq_{\mathrm{us}}))&= a + (1-a) \pde{\psi}{\xi}(\xi_{\mathrm{us}}^\ast,u) > a + 1-a = 1.
\end{align*}
Therefore, $\xeq_{\mathrm{s}}$ is locally asymptotically stable, while $\xeq_{\mathrm{us}}$ is unstable.
\end{proof}

\begin{remark}
    Observe that, for any equilibrium pair $(x,u)$ with $x \in \mathcal{X}^n$ and $u \in \U_+$, the smallest eigenvalue of the Jacobian of \eqref{eqn:model} is always greater than $-1$. 
    Consequently, no eigenvalue can cross $-1$, and hence no flip (period-doubling) bifurcation can occur \cite{kuznetsovElementsAppliedBifurcation2023}.
    A sketch of the proof is as follows. By Weyl's theorem \cite[Thm 4.3.1]{hornMatrixAnalysis2012}, it holds that $\lambda_1(J(x)) \ge a \lambda_1(W)+(1-a)\lambda_1(\Psi_x(x))$. 
    Since $W$ is row-stochastic, Perron-Frobenius arguments imply $\lambda_1(W)>-1$. Moreover, by Lemma~\ref{lem:location_eq}, any equilibrium satisfies $x_i\le u$ component-wise; by construction, this implies that $\pde{\psi_i}{x_i}(x_i,u)\ge0$ given $x_i\le u$ for all $i$. Hence, $\lambda_1(J(x)) \ge a \lambda_1(W) \ge -a > -1$.
\end{remark}

The theoretical results established in this section, illustrated by the bifurcation analysis shown Fig.~\ref{fig:bifurcation}, provide the following interpretation of time-invariant broadcast recommendations. 
When the recommendation value is low ($u<\uast$), at equilibrium, the distance to the target (i.e., $\|\xtarget-\xeq\|$) remains small. However, when the recommendation value is high ($u > \uast$), a nontrivial collective behavior emerges. If the initial conditions are sufficiently close to the target, the distance between the collective opinion at equilibrium (captured by the consensus equilibrium point $\xeq=u\1$) and the target decreases as the recommendation increases. Otherwise, for general initial conditions, the distance between the collective opinion (captured by the consensus equilibrium point $\xeq_\mathrm{s}$) and the target grows with the recommendation. As a result, maintaining a constant high broadcast recommendation can lead to worse outcomes than a strategy with a low recommendation value.

\section{Time-varying targeted recommendations}\label{sec:greedy}

After analyzing constant broadcast policies and their limitations, we now design time-varying targeted recommendations. We first derive a greedy optimal policy, which we show to be also globally optimal (Section~\ref{sec:greedy_targeting}). Then we extend the analysis to a robust policy for uncertain sensitivity $\alpha$ (Section~\ref{sec:robust-alpha}). 

\subsection{Greedy targeting}\label{sec:greedy_targeting}
The optimal targeting controller selects $u_i(t)$, $i=1,\dots,n$, such that the agents achieve convergence to the desired state $\xtarget=\1$ in the shortest time. As a first step of the analysis, we determine how to bring the agents closer to $\xtarget=\1$ within one timestep. We do this by maximizing the influence term described by the nonlinearity.
\begin{lemma}[Pointwise Optimal Control]\label{lem:pointwise-optimal}
For each $i=1,\dots,n$ and $x_i\in\mathcal{X}$, the function $u_i\mapsto \psi_i(x_i,u_i)$ over $u_i\in\U$ is maximized at
\begin{align}
\opt{u}_i(x_i) = \min\left\{1,\; \tfrac{x_i}{2} + \sqrt{\tfrac{x_i^2}{4} + \tfrac{1}{2\alpha}}\right\}.
\label{eqn:optimal-control}
\end{align}
\end{lemma}
\begin{proof}
Setting $\frac{\partial \psi_i}{\partial u_i} = \exp(-\alpha(u_i-x_i)^2)(1 - 2\alpha u_i(u_i-x_i)) = 0$ yields $2\alpha u_i^2 - 2\alpha x_i u_i - 1 = 0$ with roots $u_{i,1} = \frac{x_i}{2} -\sqrt{\frac{x_i^2}{4} + \frac{1}{2\alpha}}$ and $u_{i,2} = \frac{x_i}{2} +\sqrt{\frac{x_i^2}{4} + \frac{1}{2\alpha}}$.
Since $u_{i,1} < 0 < u_{i,2}$ and $\psi_i(x_i, u_{i,2}) > 0 > \psi_i(x_i, u_{i,1})$, the unconstrained maximum is at $u_{i,2}$. From $\psi(x_i,u)\to 0$ for $x_i$ fixed and $u\to\pm\infty$, we see that $u_{i,2} > 0 > -1$ is the global maximum and $\psi_i(x_i,\cdot)$ is increasing on $(u_{i,1}, u_{i,2})$ between the extrema. Thus, constraining to $\U$ gives~\eqref{eqn:optimal-control}.
\end{proof}

With the resulting optimal input \eqref{eqn:optimal-control}, we define the influence from the optimal recommendation $\opt{\psi}_i(x_i) := \psi_i(x_i, \opt{u}_i(x_i))$ and show that it is monotonically increasing.

\begin{lemma}[Monotonicity of Optimal Response]\label{lem:monotone-response}
Each function $\opt{\psi}_i:\mathcal{X}\to[0,1]$, $i=1,\dots,n$, $\opt{\psi}_i(x_i) := \psi_i(x_i, \opt{u}_i(x_i))$ is strictly increasing.
\end{lemma}
\begin{proof}
Let $x'$ be the threshold where the constraint activates ($\opt{u}_i(x')=1$). In the unconstrained region ($x_i < x'$), the optimality condition $1 - 2\alpha u_{i,1}(u_{i,2} - x_i) = 0$ gives $\frac{d\opt{\psi}_i}{dx_i} = \exp(-\alpha(x_i - u_{i,2})^2) > 0$, with $u_{i,2}$ from Lemma~\ref{lem:pointwise-optimal}. In the constrained region ($x_i > x'$), $\opt{\psi}_i(x_i) = \exp(-\alpha(x_i-1)^2)$, hence $\frac{d\opt{\psi}_i}{dx_i} = 2\alpha(1-x_i)\exp(-\alpha(x_i-1)^2) > 0$ for $x_i < 1$.
\end{proof}

\begin{remark}
    When deriving an optimal controller, we usually also consider the \emph{control effort} through a cost of the input.
    Since here we interpret the control input as a recommendation, e.g., which content to show in a recommender system, there is no motivation for a cost on the input when deriving the optimal controller.
\end{remark}

Now we analyze what happens when we apply the greedy one-step policy \eqref{eqn:optimal-control} repeatedly. We show that the greedy policy is optimal for any monotone objective function $c$.

\begin{theorem}[Greedy Optimality for Monotone Objectives]\label{thm:greedy-monotone}
Let the cost $c:\R^n\to\R$ be monotonically increasing (i.e., $x \leq y$ component-wise $\Rightarrow c(x) \leq c(y)$). Then, the greedy policy $u(t) = \opt{u}(x(t))$ defined in \eqref{eqn:optimal-control} solves $\max_{u(0),\ldots,u(T-1)} c(x(T))$ subject to~\eqref{eqn:model}.
\end{theorem}
\begin{proof}
We show by induction that the greedy trajectory $x'(t)$ dominates any other trajectory $x(t)$ component-wise for all $t$. The base case is immediate ($x'(0)=x(0)$). For the inductive step, assuming $x'(t) \geq x(t)$ and knowing that $w_{ij} \geq 0$ we have $Wx'(t) \geq Wx(t)$.

By Lemma~\ref{lem:pointwise-optimal}, $\opt{\psi}_i(x_i(t)) \geq \psi_i(x_i(t), u_i)$ for any $u_i \in \U$ and  by Lemma~\ref{lem:monotone-response}, $\opt{\psi}_i(x'_i(t)) \geq \opt{\psi}_i(x_i(t))$ since $x'(t) \geq x(t)$. Together this gives $\psi_i(x'_i(t), \opt{u}_i(x'_i(t))) \geq \psi_i(x_i(t), u_i(t))$.

Since each term in the convex combination $x'(t+1)$ dominates the corresponding in $x(t+1)$, we obtain $x'(t+1) \geq x(t+1)$. Monotonicity of $c$ then gives $c(x'(T)) \geq c(x(T))$.
\end{proof}

This result directly yields that the greedy policy \eqref{eqn:optimal-control} minimizes the distance to the objective $\xtarget=\1$.

\begin{corollary}[Greedy Optimality]\label{cor:distance-optimal}
The greedy policy~\eqref{eqn:optimal-control} minimizes $\|x(T) - \1\|_2^2$.
\end{corollary}
\begin{proof}
$c(x) = -\|x - \1\|_2^2 = -\sum_i (x_i - 1)^2$ is monotone increasing on $\mathcal{X}^n$ since $\frac{\partial}{\partial x_i}(-(x_i-1)^2) = 2(1-x_i) \geq 0$. The result follows directly from Theorem~\ref{thm:greedy-monotone}.
\end{proof}

As a last step in the analysis of the targeted policies, the following theorem establishes convergence to $\xtarget=\1$ from any initial condition $x\in\mathcal{X}^n$.

\begin{theorem}[Convergence to Target]\label{thm:convergence}
Under the greedy policy \eqref{eqn:optimal-control}, $x(t) \to \xtarget=\1$ for $t\to\infty$ from any $x(0)\in\mathcal{X}^n$.
\end{theorem}
\begin{proof}
    Since $\psi_i(x_i, x_i) = x_i$ and $u = x_i$ does not meet the optimality condition $1 - 2\alpha u(u - x_i) = 0$, the maximizer must satisfy $\opt{u}_i(x_i) \neq x_i$ (for $x_i\neq1$). We obtain the inequality
    \begin{equation}
        \opt{\psi}_i(x_i) > x_i\text{ for all }x_i\in [-1,1).
        \label{eqn:psigeqx}
    \end{equation}
    Define $m(t) := \min_i x_i(t)$. Since $W$ is row-stochastic, $\sum_j w_{ij} x_j(t) \geq m(t)$. Combined with \eqref{eqn:psigeqx}, we obtain $x_i(t+1) \geq a\, m(t) + (1-a)\opt{\psi}(x_i(t)) > a\, m(t) + (1-a) m(t) = m(t)$ whenever $x_i(t) < 1$. Since this holds for all $i$ (and trivially if $x_i(t)=1$), we obtain, for  $m(t) < 1$, that
    \begin{equation}
        m(t+1) \geq a\, m(t) + (1-a)\opt{\psi}(m(t)) > m(t).
        \label{eqn:m_inequality}
    \end{equation}
    
    Since $\{m(t)\}$ is strictly increasing and bounded above by~1, it converges to some $m^* \leq 1$. Taking the limit in \eqref{eqn:m_inequality} and using the continuity of $\opt{\psi}$ yields $m^* \geq a\, m^* + (1-a)\opt{\psi}(m^*)$, i.e., $m^* \geq \opt{\psi}(m^*)$. By \eqref{eqn:psigeqx}, this is impossible unless $m^* = 1$. Hence $\min_i x_i(t) \to 1$ as $t\to\infty$, which implies $x(t) \to \1=\xtarget$ as $t\to\infty$.
\end{proof}

In summary, we have shown that for the time-varying targeting case, a greedy strategy is optimal and yields convergence of the trajectories of \eqref{eqn:model} to the desired state $\xtarget$. Since this policy only depends on the state of the individual agent, it would be possible to implement these personal recommendations in a distributed way, i.e., everyone that can observe the state of an individual $x_i$ and knows the sensitivity $\alpha$ can give optimal recommendations. This means that the recommendations can be given by distributed actors that only need to communicate global parameters (here $\alpha$) rather than a centralized instance with full knowledge of the system. In the next section, we will provide robust recommendations for the case of uncertain sensitivity.

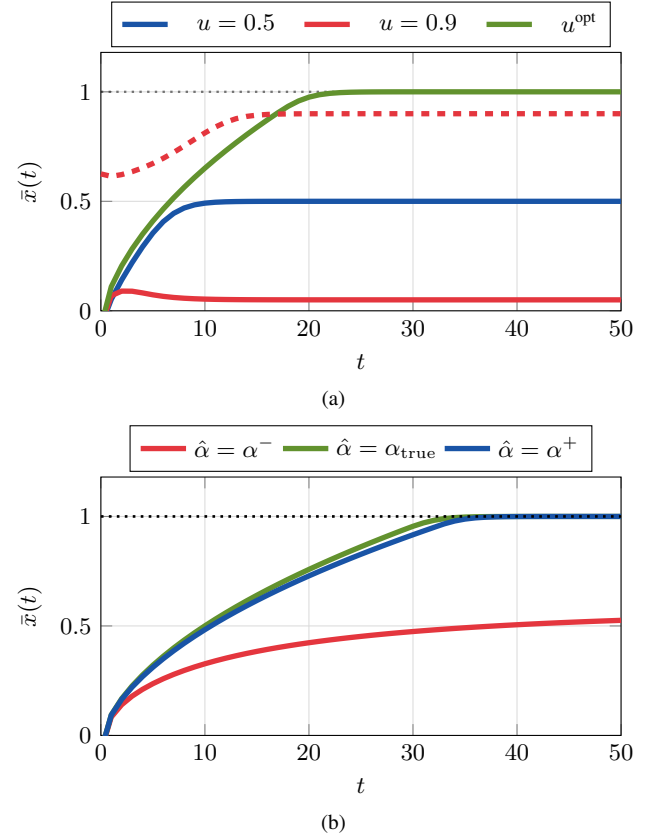
\begin{figure}[ht]\centering
    \subfloat[]{%
        \begin{minipage}{0.98\linewidth}
            \pgfplotsset{width=\linewidth}%
            \begin{tikzpicture}
\begin{axis}[
  width=\linewidth,
  height=5cm,
  xlabel={$t$},
  ylabel={$\bar{x}(t)$},
  xmin=0, xmax=50,
  ymin=0, ymax=1.18,
  legend style={
    at={(0.5,1.02)},
    anchor=south,
    legend columns=-1,
    font=\small,
    column sep=0.8em,
  },
  grid=major,
  grid style={line width=0.2pt, draw=gray!30},
  tick label style={font=\small},
  label style={font=\small},
]

%% Target line
\addplot[forget plot, black, dotted, line width=0.9pt, opacity=0.5, domain=0:50, samples=2] {1};

%% Broadcast u=0.5 — mean
\addplot[kth-blue, line width=2pt]
  table[x=t, y=mean, col sep=comma]{results/exp1_bc05.csv};
\addlegendentry{$u=0.5$}

%% Broadcast u=0.9 (low IC) — mean
\addplot[kth-lightred, line width=2pt]
  table[x=t, y=mean, col sep=comma]{results/exp1_bc09.csv};
\addlegendentry{$u=0.9$}

%% Greedy optimal — mean
\addplot[kth-green, line width=2pt]
  table[x=t, y=mean, col sep=comma]{results/exp1_opt.csv};
\addlegendentry{$\opt{u}$}

%% Broadcast u=0.9 (high IC) — mean
\addplot[forget plot, kth-lightred, line width=2pt, dashed]
  table[x=t, y=mean, col sep=comma]{results/exp1_bc09_high.csv};

\end{axis}
\end{tikzpicture}
        \end{minipage}
        \label{fig:results_combined-a}}
    \\ \vspace{-.5em}
    \subfloat[]{%
        \begin{minipage}{0.98\linewidth}
            \pgfplotsset{width=\linewidth}%
            \begin{tikzpicture}
\begin{axis}[
  width=\linewidth,
  height=5cm,
  xlabel={$t$},
  ylabel={$\bar{x}(t)$},
  xmin=0, xmax=50,
  ymin=0, ymax=1.18,
  legend style={
    at={(0.5,1.02)},
    anchor=south,
    legend columns=-1,
    font=\small,
  },
  grid=major,
  grid style={line width=0.2pt, draw=gray!30},
  tick label style={font=\small},
  label style={font=\small},
]

%% alpha_lo
\addplot[kth-lightred, line width=2pt]
  table[x=t, y=lo_mean, col sep=comma]{results/exp2.csv};
\addlegendentry{$\hat\alpha=\alpha^-$}

%% alpha_true
\addplot[kth-green, line width=2pt]
  table[x=t, y=true_mean, col sep=comma]{results/exp2.csv};
\addlegendentry{$\hat\alpha=\alpha_{\mathrm{true}}$}

%% alpha_hi (robust)
\addplot[kth-blue, line width=2pt]
  table[x=t, y=hi_mean, col sep=comma]{results/exp2.csv};
\addlegendentry{$\hat\alpha=\alpha^+$}

%% Target
\addplot[black, dotted, line width=1pt, forget plot]
  coordinates {(0,1)(50,1)};
\addlegendentry{$x^*=1$}

\end{axis}
\end{tikzpicture}
        \end{minipage}
        \label{fig:results_combined-b}}
    \caption{Mean state $\bar{x}(t)=1/n\sum_{i=1}^{n}x_i(t)$ for different closed-loop trajectories. The controller should steer the state towards $\xtarget=\1$.  (a) Comparison of the homogeneous broadcast (blue and red trajectories) with the greedy policy \eqref{eqn:optimal-control} (green trajectory). (b) Robust version of the greedy policy for $\alpha_{\mathrm{true}}\in[\alpha^-,\alpha^+]=[3,10]$ and $\alpha_{\mathrm{true}}=7$.}\vspace{-0.3cm}
    \label{fig:results_combined}
\end{figure}

\subsection{Robust greedy policy}\label{sec:robust-alpha}

The sensitivity $\alpha$ describes how large the acceptance window for recommendations is: a smaller $\alpha$ leads to a larger acceptance window and thus to a more aggressive optimal recommendation $\opt{u}$. Suppose $\alpha$ is only known to lie in an interval $[\alpha^-,\alpha^+]$, with $0<\alpha^-\leq\alpha^+$. We show that designing the greedy controller with $\hat\alpha=\alpha^+$ guarantees convergence for any true value of the sensitivity $\alpha_{\mathrm{true}}\in[\alpha^-,\alpha^+]$.

\begin{theorem}[Robust Convergence]\label{thm:robust-convergence}
Let $\alpha_{\mathrm{true}}\in[\alpha^-,\alpha^+]$ be unknown. Under the greedy controller designed with $\hat\alpha=\alpha^+$, i.e.,\vspace{-0.1cm}
\begin{align}
    \opt{u}_i(x_i) = \min\!\left\{1,\; \frac{x_i}{2} + \sqrt{\frac{x_i^2}{4} + \frac{1}{2\alpha^+}}\right\}, \label{eqn:robust-control}
\end{align}
the system~\eqref{eqn:model} converges to $\xtarget=\1$ from any $x(0)\in\mathcal{X}^n$.
\end{theorem}

\begin{proof}
Denote by $u^{\alpha}_i(x_i)$ ($\psi_i^{\alpha}$) be the optimal input (influence function)  calculated with $\alpha$. 
By the proof of Theorem~\ref{thm:convergence}, it suffices to verify $\psi_i^{\alpha_{\mathrm{true}}}(x_i,u^{\alpha^+}_i(x_i))>x_i$ for all $x_i\in[-1,1)$. Since $\alpha_{\mathrm{true}}\leq\alpha^+$,
\begin{align*}
    e^{-\alpha_{\mathrm{true}}(u^{\alpha^+}_i-x_i)^2} \geq e^{-\alpha^+(u^{\alpha^+}_i-x_i)^2}
\end{align*}
and with $u^{\alpha^+}_i>0$ we have $\psi_i^{\alpha_{\mathrm{true}}}(x_i,u^{\alpha^+}_i)\geq \psi_i^{\alpha^+}(x_i,u^{\alpha^+}_i)$. By Lemma~\ref{lem:pointwise-optimal}, $u^{\alpha^+}_i$ maximizes $\psi_i^{\alpha^+}(x_i,\cdot)$, and since $u^{\alpha^+}_i\neq x_i$ for $x_i<1$, we have $\psi_i^{\alpha^+}(x_i,u_i)> \psi_i^{\alpha^+}(x_i,x_i)=x_i$. The result follows from Theorem~\ref{thm:convergence}.
\end{proof}

The choice $\hat\alpha=\alpha^+$ assumes the narrowest acceptance window for recommendations and any $\alpha_{\mathrm{true}}<\alpha^+$ only makes agents more receptive to recommendations than assumed by the robust policy. Conversely, choosing $\hat\alpha<\alpha^+$ may fail because the controller recommends $u_i$ too far from $x_i$. The trade-off for guaranteeing robust convergence is slower convergence when $\alpha_{\mathrm{true}}\ll\alpha^+$, since the controller is more conservative than necessary.

\section{Numerical example}\label{sec:example}

The numerical example in this section illustrates the theoretical findings presented in Sections~\ref{sec:broadcast} and \ref{sec:greedy}. We consider a connected graph $\G(W)$ with $n=20$ agents, where each edge exists with probability $p=0.2$, and the edge weights are drawn from a uniform distribution and normalized such $W\1=\1$.
The parameters of the model \eqref{eqn:model} are $a=0.5$ and $x(0)$ uniformly drawn from $\mathcal{X}^n$. The goal is for the recommendations to steer the opinions towards $\xtarget=\1$. 

Fig.~\ref{fig:results_combined} shows two experiments performed on the same underlying social network. In Fig.~\ref{fig:results_combined-a} we compare the homogeneous broadcast to the greedy policy \eqref{eqn:optimal-control}. This example shows that the initial opinions have to be sufficiently biased for the broadcast recommendation to achieve convergence to $\xeq=u\1$ (as detailed in Section~\ref{sec:broadcast}). For smaller $u=0.5<\uast$ the system converges to $\xeq=u\1$. Due to the fold bifurcation behavior (see Theorems~\ref{thm:existence} and \ref{thm:stability}), a large $u$ (here $u=0.9>\uast$) can lead to different equilibrium points (compare the red and dashed red lines), depending on the initial condition: The agents can move in the opposite direction of what is recommended. In contrast, the greedy policy derived in Lemma~\ref{lem:pointwise-optimal} ensures convergence by giving optimal personal recommendations to each agent in the network. In Fig.~\ref{fig:results_combined-b}, we analyze the robust version of the greedy policy for an uncertain sensitivity $\alpha$ (Theorem~\ref{thm:robust-convergence}). The robust $\hat{\alpha}=\alpha^+$ leads to slower convergence than perfect estimation $\hat{\alpha}=\alpha_{\mathrm{true}}$. However, choosing $\hat{\alpha}$ too small (here $\hat{\alpha}=\alpha^-$) can prevent the system from converging to $\xtarget$.

\section{Conclusion}\label{sec:conclusion}
This work addressed the issue of personal recommendations in social networks. We proposed a nonlinear interconnected model characterized by a smooth nonlinear Gaussian influence function, consistent with effects from behavioral psychology, that describes how agents process recommendations.
We analyzed our model under broadcast and targeted policies. In the broadcast scenario, our bifurcation analysis shows that, depending on the initial opinions, it is not always possible to reach the desired target of the recommendations. We further designed optimal personal (targeted) recommendations and provided results for robust convergence when the agents' sensitivities are uncertain. Our model is analytically tractable and can serve as a foundation for analyzing personal recommendations.

In future work, we plan to extend the model to heterogeneous sensitivities $\alpha_i$ and weighting of the social term $a_i$, as well as realistic opinion-forming processes that capture, for instance, the agents' biases and stubbornness.

\appendices
\section{Properties of the nonlinear function}\label{sec:propertiesNonlinearFunctions}
We collect some properties of the nonlinearity $\psi_i=\psi_i(x_i,u_i)$ in \eqref{eqn:model}. The first and second derivatives are
    $\pde{\psi_i}{x_i}(x_i,u_i) = 2 \alpha u_i (u_i-x_i)  e^{-\alpha (u_i-x_i)^2}$
    and
    $\ppde{\psi_i}{x_i}(x_i,u_i)%
    = 2 \alpha u_i (- 1 + 2 \alpha (u_i-x_i)^2 ) e^{-\alpha (u_i-x_i)^2}$.
The Lipschitz constant is given by $L_{\psi}(u_i)= \sqrt{2\alpha}  e^{-1/2}\, u_i$.

Assume that $u_i$ is fixed, i.e., $u_i=\uast\in \U$. Then $\psi_i(x_i,\uast)$ has a maximum value equal to $\uast$ at $x_i=\uast$, i.e.,  $\psi_i(\uast,\uast)=\uast$ and $\pde{\psi_i}{x_i}(\uast,\uast)=0$. The inflection points are given by $x_i =  \uast \pm \frac{1}{\sqrt{2\alpha}}$, yielding 
$\psi_i(\uast\pm\tfrac{1}{\sqrt{2\alpha}},\uast)= \uast e^{-\frac{1}{2}}
$, and 
$\pde{\psi_i}{x_i}(\uast\pm \tfrac{1}{\sqrt{2\alpha}},\uast) =  \mp \sqrt{2\alpha} e^{-\frac{1}{2}}\,\uast $.

\bibliographystyle{IEEEtran}

\end{document}